\documentclass[aps,pra,twocolumn,nofootinbib,10pt]{revtex4-1}
 \pdfoutput=1

\usepackage{graphicx}
\usepackage{amsmath}
\usepackage{amssymb}
\usepackage{mathrsfs}
\usepackage{amsthm}
\usepackage{bm}
\usepackage{url}
\usepackage[T1]{fontenc}
\usepackage{csquotes}
\MakeOuterQuote{"}

\newtheoremstyle{note}
  {\topsep/2}               
  {\topsep/2}               
  {}                      
  {\parindent}            
  {\itshape}              
  {.}                     
  {5pt plus 1pt minus 1pt}
  {}

\theoremstyle{note}
\newtheorem{theorem}{Theorem}
\newtheorem{lemma}{Lemma}

\newtheorem{corollary}{Corollary}

\theoremstyle{definition}

\theoremstyle{remark}
\newtheorem{remark}{Remark}

\newcommand{\mrm}[1]{\mathrm{#1}}
\newcommand{\tr}{\operatorname{tr}}

\newcommand{\rmi}{\mathrm{i}}
\newcommand{\rme}{\mathrm{e}}

\newcommand{\rmd}{\mathrm{d}}
\newcommand{\rmT}{\mathrm{T}}

\newcommand{\be}{\begin{equation}}
\newcommand{\ee}{\end{equation}}
\newcommand{\ba}{\begin{align}}
\newcommand{\ea}{\end{align}}

\def\<{\langle}  
\def\>{\rangle}  

 \newcommand{\caH}{\mathcal{H}}

\newcommand{\bbF}{\mathbb{F}}

\newcommand{\bbZ}{\mathbb{Z}}
\newcommand{\Sp}[2]{\mrm{Sp}(#1,#2)}

\newcommand{\SL}[2]{\mrm{SL}(#1,#2)}

\newcommand{\phw}{\overline{D}}

\newcommand{\pc}{\overline{\mathrm{C}}}

\newcommand{\pcr}{\overline{\mathrm{C}}_{\mrm{r}}}

\def\eqref#1{\textup{(\ref{#1})}}  
\newcommand{\eref}[1]{Eq.~\textup{(\ref{#1})}}
\newcommand{\Eref}[1]{Equation~\textup{(\ref{#1})}}
\newcommand{\esref}[1]{Eqs.~\textup{(\ref{#1})}}

\newcommand{\thref}[1]{Theorem~\ref{#1}}
\newcommand{\Thref}[1]{Theorem~\ref{#1}}
\newcommand{\thsref}[1]{Theorems~\ref{#1}}

\newcommand{\lref}[1]{Lemma~\ref{#1}}

\newcommand{\cref}[1]{Conjecture~\ref{#1}}
\newcommand{\Cref}[1]{Conjecture~\ref{#1}}

\newcommand{\rcite}[1]{Ref.~\cite{#1}}
\newcommand{\rscite}[1]{Refs.~\cite{#1}}

\begin{document}

\title{Multiqubit Clifford groups are unitary 3-designs}
\author{Huangjun Zhu}
\affiliation{Perimeter Institute for Theoretical Physics, Waterloo, On N2L 2Y5, Canada}
\affiliation{Institute for Theoretical Physics, University of Cologne, 
Cologne 50937, Germany}
\email{zhuhuangjun@fudan.edu.cn}
\email{zhuhuangjun@gmail.com}

\begin{abstract}
Unitary $t$-designs are a ubiquitous tool in many research areas, including randomized benchmarking, quantum process tomography, and scrambling. Despite the intensive efforts of many researchers, little is known about unitary $t$-designs with $t\geq3$ in the literature. We show that the multiqubit  Clifford group in any even prime-power dimension is not only a unitary 2-design, but also a  3-design. Moreover, it is a minimal  3-design except for dimension~4.
As an immediate consequence, any orbit of pure states of the multiqubit Clifford group forms a complex projective 3-design; in particular, the set of stabilizer states forms a 3-design. In addition, our study is  helpful to studying higher moments of the Clifford group, which are useful  in many research areas ranging from quantum information science to signal processing. Furthermore, we  reveal a surprising connection between unitary 3-designs and the physics of discrete phase spaces and thereby offer a simple explanation of why no discrete Wigner function is covariant with respect to the multiqubit Clifford group,  which is of intrinsic interest to studying quantum computation.
\end{abstract}

\date{\today}
\maketitle

\section{Introduction}

Unitary designs are a ubiquitous tool in quantum information science \cite{DiViLT02, Chau05,Dank05the, DankCEL09, GrosAE07, RoyS09, ClevLLW16}. They are particularly useful in derandomizing constructions that rely on random unitaries, such as
randomized benchmarking \cite{KnilLRB08,MageGE11,WallF14},   quantum process tomography \cite{Scot08,KimmL17},   quantum cryptography \cite{Chau05,AmbaBW09}, and data hiding \cite{DiViLT02}. In addition, they can generate complex projective designs \cite{ReneBSC04, Scot06, AmbaE07}, which are equally useful in derandomizing constructions that rely on random quantum states. Recently, projective and unitary designs have also found increasing applications beyond quantum information science, especially in the study of chaos and scrambling  \cite{HaydP07,SekiS08,HosuQRY16,RobeY17}. 

Most previous studies on this subject have focused on unitary 2-designs, among which the Clifford group is the most prominent \cite{Gott97the, Chau05,Dank05the,DankCEL09, Gros06, GrosAE07,KnilLRB08,MageGE11,WallF14, ClevLLW16} due  to its extensive applications in various research areas, such as quantum computation, quantum error correction, and randomized benchmarking. Complex projective 2-designs constructed from Clifford orbits, including the set of stabilizer states in particular, are also of special interest \cite{Gott97the, Gros06,HowaWVE14}. 
By contrast, little is known about  $t$-designs with $t\geq3$ except for randomized constructions \cite{AmbaE07,HarrL09, BranHH16, CwikHMP13, NakaHKW17}, despite the intensive efforts of many researchers in the past decade. This situation  has set a big barrier in realizing many tasks that rely on higher $t$-designs, such as quantum state discrimination \cite{AmbaE07,MattWW09}, quantum  tomography  \cite{HayaHH05,GrosLFB10, KimmL17}, phase retrieval  \cite{GrosKK15,KuenRT17}, and reduction of query complexity \cite{BranH13}.

Here we show that
the multiqubit (including single-qubit) Clifford group  is not only a unitary 2-design, but also  a 3-design. Moreover, it is  \emph{minimal} except for dimension~4 in the sense that it does not contain any proper subgroup that is also a unitary 3-design.
As a consequence, any orbit of pure states of the multiqubit Clifford
group, including the set of stabilizer states in particular, forms a 3-design, which extends the result in \rcite{KuenG13}. Our study not only  provides   infinite families of well-structured  3-designs, but also 
paves the way for constructing  $t$-designs with even higher strengths \cite{ZhuKGG16}. Recently, these results have  found satisfactory applications in  quantum state discrimination \cite{KuenZG16D} and phase retrieval~\cite{KuenZG16L}. Furthermore,  our work is helpful to  studying multipartite entanglement in stabilizer tensor networks, which stand as an effective tool for understanding holographic  duality \cite{NezaW16}.

In addition, our
study  leads to a simple explanation of the distinction
between discrete Wigner functions in even prime-power dimensions and those in odd prime-power dimensions \cite{Woot87,Gros06,Zhu16P}. 
This distinction  has been an elusive question and has profound implications for various interesting subjects, such as  computational speedup and contextuality \cite{VeitFGE12, HowaWVE14,DelfABR15}. In each odd prime-power dimension, the discrete Wigner function introduced by Wootters \cite{Woot87} is covariant with respect to the Clifford group \cite{Gros06,Zhu16P}; by contrast,  none is covariant with respect to the multiqubit Clifford group \cite{Zhu16P}. Here we reveal a surprising connection between unitary 3-designs and the physics of discrete phase spaces and thereby clarify the reason behind this distinction.

\section{Preliminaries}

A set of pure quantum states $\{|\psi_j\rangle\}$ in a $d$-dimensional Hilbert space $\caH$  is a (complex projective) \emph{t-design}  for a positive integer $t$ if $\sum_j (|\psi_j\rangle\langle \psi_j|)^{\otimes t}$ is proportional to the projector onto the symmetric subspace of  $\caH^{\otimes t}$ \cite{AmbaE07, ReneBSC04, Scot06, ApplFZ15G}. A set of $K$ unitary operators $\{U_j\}$ acting on $\caH$ is a  \emph{unitary $t$-design} \cite{Dank05the, DankCEL09, GrosAE07} if it satisfies
\begin{equation}\label{eq:U2design}
\frac{1}{K} \sum _j U_j^{\otimes t} M(U_j^{\otimes t})^\dag =\int \rmd U U^{\otimes t}M(U^{\otimes t})^\dag
\end{equation}
for any linear operator $M$ acting on $\caH^{\otimes t}$, where $\dag$ stands for the Hermitian conjugate and
the integral is taken  with respect to the normalized Haar measure. By definition, a unitary $t$-design is also a $t^\prime$-design
for $t^\prime<t$. Note that the above equation remains intact when $U_j$ are multiplied by arbitrary phase factors, so what we are concerned with are actually projective unitary $t$-designs.  Alternatively, the set $\{U_j\}$ is a unitary $t$-design if the $t$th \emph{frame potential}
\begin{equation}
\Phi_t(\{U_j\}):=\frac{1}{K^2}\sum_{j,k} |\tr(U_jU_k^\dag)|^{2t}
\end{equation}
attains the minimum $\gamma(t,d):=\int \rmd U |\tr(U)|^{2t}$ \cite{GrosAE07, Scot08, RoyS09}.
Here we only need  $\gamma(t,d)$ in  two special cases \cite{Rain98,Scot08},
\begin{equation}\label{eq:FramePmin}
\gamma(t,d)=\begin{cases}
\frac{(2t)!}{t! (t+1)!}, &d=2,\\
t!,& d\geq t.
\end{cases}
\end{equation}
Besides the current application,  frame potentials  
also play an important role in  studying chaos and circuit complexity \cite{HaydP07,SekiS08,HosuQRY16,RobeY17}.

Most known examples of unitary designs are constructed from subgroups of the unitary group, which are referred to as (unitary) group designs.
Given a finite group~$G$ of unitary operators on $\caH$, in most cases we are only concerned with the quotient $\overline{G}$ of $G$ over the phase factors.
The frame potential of  $\overline{G}$ takes on the form \cite{GrosAE07}
\begin{equation}
\Phi_t(\overline{G}):=\frac{1}{|\overline{G}|}\sum_{U\in \overline{G}} |\tr(U)|^{2t}=\frac{1}{|G|}\sum_{U\in G} |\tr(U)|^{2t},
\end{equation}
where $|\overline{G}|$ and $|G|$ denote the orders of $\overline{G}$ and $G$. 
Note that $\Phi_t(\overline{G})$ coincides with  the sum of squared multiplicities of irreducible components of $\overline{G}_{(t)}:=\{U^{\otimes t}|U\in \overline{G}\}$. The group $\overline{G}$ is a unitary $t$-design if and only if $\overline{G}_{(t)}$ has the same number of irreducible components as $\mathrm{U}_{(t)}$, where $\mathrm{U}$ denotes the group of all unitary operators acting on $\caH$~\cite{GrosAE07}.
For example, the group  $\overline{G}$ is a unitary 1-design if and only if it is irreducible. 
It is a  2-design if  $\overline{G}_{(2)}$ has two irreducible components, which correspond to the symmetric subspace and antisymmetric subspace of $\caH^{\otimes2}$. Prominent examples of unitary group 2-designs include Clifford groups and restricted Clifford groups in prime-power dimensions \cite{DiViLT02, Chau05, Dank05the, DankCEL09, GrosAE07}.  Not much is known  about unitary $t$-designs with larger  $t$. 

Before presenting our main results, we need to introduce
the (multipartite) \emph{Heisenberg-Weyl} (HW) group.
In prime dimension $p$,  the HW group $D$ is
generated by the phase operator $Z$
 and the cyclic-shift operator $X$,
\begin{equation}\label{eq:HW}
Z|u\rangle=\omega^u |u\rangle, \quad X|u\rangle=|u+1\rangle,
\end{equation}
where $\omega=\mathrm{e}^{2\pi \mathrm{i}/p}$,
$u\in \bbF_p$, and $\bbF_p$  is the field
of integers modulo $p$ (often written as $\bbZ_p$). When $p=2$, the operators $Z$
and  $X$ reduce to the familiar Pauli operators $\sigma_z$ and $\sigma_x$, and the HW group reduces to the Pauli group. In this case, it is often more convenient to consider a variant of  the HW group which includes the scalar $\rmi$ as an additional generator. However, these choices do not affect the following discussion since we are mostly concerned with the HW group modulo phase factors.

In prime-power dimension $q=p^n$, the  HW group $D$ is the tensor power of $n$ copies of the HW group in  dimension $p$. The elements in the HW group are called displacement operators (or Weyl operators). Up to phase factors, they can be labeled by  vectors  in $\mathbb{F}_p^{2n}$ as 
\begin{equation}
D_{\mu}= \tau^{\sum_j \mu_j \mu_{n+j} }\prod_{j=1}^n X_j^{\mu_j}  Z_j^{\mu_{n+j}},
\end{equation}
where  $\tau=-\rme^{\pi \rmi/p}$,  while $Z_j$ and $X_j$ are the phase operator and cyclic shift operator of the $j$th party. These operators satisfy the commutation relation
\begin{equation}
D_{\mu}D_{\nu} D_{\mu}^\dag D_{\nu}^\dag  =\omega^{\langle\mu,\nu\rangle},
\end{equation}
where $\langle\mathbf{\mu},\mathbf{\nu}\rangle=\mu^\rmT J\nu$ is the symplectic product with
$J=\bigl(\begin{smallmatrix}0_n &-1_n\\ 1_n& 0_n
\end{smallmatrix}\bigr)$.
  The symplectic group $\Sp{2n}{p}$ is the group of linear transformations on $\mathbb{F}_p^{2n}$ that preserves the symplectic product.

The (full) \emph{Clifford group} $\pc$ is  composed of all unitary transformations that map displacement operators to displacement operators up  to phase factors \cite{BoltRW61I,BoltRW61II,Gott97the,Gros06,Zhu15Sh,Zhu16P}. It  is referred to as the multiqubit Clifford group when the dimension is a power of 2 (including 2). In dimension 2, the Clifford group $\pc$ corresponds to the symmetry group of a cube inscribed in the Bloch sphere. In general, any Clifford unitary $U$ induces a symplectic transformation $F$ on the symplectic space $\mathbb{F}_p^{2n}$ that labels the displacement operators. Conversely, given any symplectic matrix $F$, there exist $q^2$ Clifford unitaries (up to phase factors) that induce  $F$ \cite{BoltRW61I,BoltRW61II,Zhu15Sh}. The quotient  $\pc/\phw$ can be identified with the symplectic group $\Sp{2n}{p}$ \cite{BoltRW61I,BoltRW61II}.

The symplectic space $\bbF_p^{2n}$ can  also be identified with a two-dimensional vector space over the field $\bbF_q$. The special linear group $\SL{2}{q}$ on this space is an extension-field-type subgroup of $\Sp{2n}{p}$.
The  \emph{restricted Clifford group} $\pcr$ (coinciding with the full Clifford group when $q$ is a prime) is the subgroup of $\pc$ whose quotient $\pcr/\phw$ corresponds to $\SL{2}{q}$; see \rscite{Gros06, GrosAE07, Appl09P, Zhu15M} for more details.

\section{Multiqubit Clifford groups are unitary 3-designs}
In this section we prove our main result that the multiqubit Clifford group is a unitary 3-design. Consequently, any orbit of the Clifford group, including the orbit of stabilizer states, forms a complex projective 3-design. To achieve this goal, we determine the frame potentials  of the Clifford group up to order 4. Furthermore, we show that, except in dimension 4, the multiqubit Clifford group contains no proper subgroup that forms a unitary  3-design. Recently,  these results have  found  applications in many research areas both within and beyond quantum information science.

\begin{theorem}\label{thm:Clifford3design}
The multiqubit Clifford group is a unitary 3-design but not a 4-design. The Clifford group in any odd prime-power dimension is only a unitary 2-design. 
The restricted Clifford group in any prime-power dimension is only a unitary 2-design  except for dimension 2.
\end{theorem}

\begin{corollary}
	Any orbit of pure states of the multiqubit Clifford group forms a 3-design; in particular, the set of multiqubit stabilizer states forms a 3-design.
\end{corollary}
The conclusion on  stabilizer states was also proved directly by  Kueng and  Gross~\cite{KuenG13}.

\Thref{thm:Clifford3design} is a simple corollary  of  \eref{eq:FramePmin} and the following lemma, which  is proved in the appendix by virtue of \lref{lem:FPformula} below.
\begin{lemma}\label{lem:CliffordFP}
In prime-power dimension $p^n$, the Clifford group $\pc$ has frame potentials
\begin{align}
\Phi_2(\pc)&=2; \label{eq:CliffordFP2}\\
\Phi_3(\pc)&=
\begin{cases}
2p+1, & n=1,\\
2p+2, & n\geq 2;
\end{cases}   \label{eq:CliffordFP3} \\
\Phi_4(\pc)&=
\begin{cases}
p^3+p^2+p+1, & n=1,\\
2p^3+2p^2+2p+1, & n= 2,\\
2(p^3+p^2+p+1), & n\geq 3.\\
\end{cases}  \label{eq:CliffordFP4}
\end{align}
The restricted Clifford group $\pcr$ has frame potentials
 \begin{align}\label{eq:RCliffordFP}
 \Phi_t(\pcr)&=\frac{q(q^{2t-4}-1)}{q^2-1}+q^{t-2}+1\quad \forall t\geq 1.
 \end{align}
\end{lemma}
It is worth pointing out that \esref{eq:CliffordFP2} to \eqref{eq:CliffordFP4} also apply to subgroups of the Clifford group whose quotients over the HW group are isomorphic to $\Sp{2m}{p^k}$ with $mk=n$ if $n$ and $p$ are replaced by  $m$ and   $p^k$, respectively. 
Besides proving \thref{thm:Clifford3design}, \lref{lem:CliffordFP}  shows that  the Clifford group in dimension 2 is quite close to a unitary 4-design.  The one in dimension 3 is quite close to a unitary 3-design; the larger the prime $p$ is, the further away the Clifford group  is from being a unitary 3-design. In addition, the frame potentials presented in \lref{lem:CliffordFP} are 
crucial to analyzing the fourth tensor power of the Clifford group  and to constructing  4-designs from Clifford orbits \cite{ZhuKGG16,HelsWW16}, which are  useful in  quantum state discrimination \cite{KuenZG16D} and phase retrieval~\cite{KuenZG16L}.  Recently, these results also found an application in  studying multipartite entanglement in stabilizer tensor networks, which are instructive to  understanding holographic  duality \cite{NezaW16}. Furthermore, our study is helpful to exploring chaos and circuit complexity \cite{HosuQRY16,RobeY17}.

The following lemma is useful not only in proving \lref{lem:CliffordFP}, but also in computing frame potentials of  subgroups of the Clifford group that contain the HW group. See the appendix for a proof. 
\begin{lemma}\label{lem:FPformula}
Suppose $\overline{G}\geq\phw$ is a subgroup of the Clifford group $\pc$ in dimension $q=p^n$ and $R=\overline{G}/\phw$ (taken as  a subgroup of $\Sp{2n}{p})$. Then
\begin{equation}\label{eq:FPformula}
\Phi_t(\overline{G})=\frac{1}{|R|}\sum_{F\in R} f(F)^{t-1},
\end{equation}
where $|R|$ is the order of $R$ and
 $f(F)$ is the number of fixed points of $F$ in  $\bbF_p^{2n}$.
Moreover,  $\Phi_t(\overline{G})$ is equal to the number of orbits of $R$ on $(\bbF_p^{2n})^{\times (t-1)}$.   The group $\overline{G}$ is a unitary 2-design if and only if $R$ is transitive on  $\bbF_p^{2n*}$.  It  is a unitary
3-design if and only if $R$ is 2-transitive  when $n=1$ and is a rank-3 permutation
group  when $n\geq2$.
\end{lemma}

\begin{remark}
$\bbF_p^{2n*}$ is the set of nonzero vectors in $\bbF_p^{2n}$. A subgroup of $\Sp{2n}{p}$ is \emph{transitive} if it can map  any nonzero vector in $\bbF_p^{2n}$ to any other and \emph{2-transitive} or  doubly transitive if it can map any ordered pair of distinct nonzero vectors to any other pair. It is a \emph{rank-3 permutation group} if it is transitive and each point stabilizer
has three orbits  on  $\bbF_p^{2n*}$ including the orbit of the fixed point \cite{DixoM96book,Came99book,Zhu15S}. The relation between transitive subgroups of the symplectic group and unitary 2-designs was  noticed previously in \rcite{GrosAE07}.
\end{remark}

In  many applications, unitary designs with fewer elements are desirable. Is there any  proper subgroup of the multiqubit Clifford group that forms a 3-design? The answer turns out to be negative except for dimension~4. The following theorem is proved in the appendix. It shows that in a sense the multiqubit Clifford group is the most economical in constructing a unitary 3-design.
\begin{theorem}\label{thm:Clifford3designM}
The multiqubit Clifford group $\pc$  is a minimal unitary 3-design except for dimension 4, in which case it has  a unique proper subgroup that is  a unitary 3-design. 
\end{theorem}
\begin{remark}In  dimension 4, $\pc/\phw\simeq\Sp{4}{2}\simeq S_6$  contains a unique subgroup that is isomorphic to $A_6$ \cite{Wils09book}, where $S_m$ and $A_m$ denote  the symmetric group and alternating group on $m$ letters. The preimage of  $A_6$ in $\pc$ is  a unitary 3-design.   
\end{remark}

\section{Applications to discrete Wigner functions}

Discrete Wigner functions are the analogs of the familiar Wigner function in the continuous scenario. They are useful in many research areas, including quantum tomography and  quantum computation.  In each odd prime-power dimension, 
the Wootters discrete Wigner function
is distinguished 
because it is covariant with respect to the Clifford group \cite{Woot87,Gros06,Zhu16P}. 
In this quasiprobability representation, Clifford transformations can be understood as permutations on the discrete phase space. In addition,  a pure state has a nonnegative Wootters discrete Wigner function if and only if it is a stabilizer state according to the discrete Hudson theorem \cite{Gros06}. In particular,
stabilizer states can be represented as probability distributions on the discrete phase space. These facts offer a simple explanation of the famous Gottesman-Knill theorem which states that stabilizer quantum computation can be efficiently simulated  classically \cite{NielC00book}. In other words, negativity in the Wootters discrete Wigner function is a necessary resource to achieve universal quantum computation \cite{VeitFGE12}. Incidentally, this negativity  is also tied to the prominent nonclassical phenomenon known as contextuality \cite{HowaWVE14}.

In the multiqubit setting, which is the most relevant to realizing practical quantum computation, however, no discrete Wigner function is covariant with respect to the Clifford group \cite{Zhu16P}. Consequently, it is more difficult to come up with a simple geometric picture that illustrates the Gottesman-Knill theorem. Also, it is more difficult to clarify the origin of computational speedup in quantum computation based on qubits. A focus of ongoing research is to understand the distinction between multiqubit systems and systems of odd local dimensions \cite{VeitFGE12, HowaWVE14,DelfABR15,Zhu16P}.

Here we  show that the nonexistence of a Clifford covariant discrete Wigner function in an even prime-power dimension
 is closely tied to the fact that the multiqubit Clifford group is a unitary 3-design. To elucidate this point, it suffices to show that no operator basis is covariant with respect to the multiqubit Clifford group, note that any Clifford covariant discrete Wigner function determines a Clifford covariant operator basis.
For example, in each odd prime-power dimension, the Wootters discrete Wigner function  \cite{Woot87} determines  the operator basis composed of phase point operators, and vice versa \cite{Gros06,Zhu16P}. Here an operator basis $\{L_j\}$ is \emph{covariant} with respect to the group $\overline{G}$ of unitary transformations if $\overline{G}$ leaves this basis invariant and acts transitively on  the basis operators. In particular, each $U\in
\overline{G}$ induces  a permutation among the basis operators.

\begin{theorem}\label{lem:u3design}
No operator basis is covariant with respect to any unitary group 3-design. No discrete Wigner function is covariant with respect to the multiqubit Clifford group.
\end{theorem}
This theorem is proved in the appendix. It
offers a simple explanation of the distinction between multiqubit systems and systems of odd local dimensions, which is of intrinsic interest to studying quantum computation.
 Moreover, it reveals a surprising connection between unitary $t$-designs and the physics of discrete phase spaces, which may have profound implications for the cross fertilization of the two active research fields.

\section{Summary}

We showed that the multiqubit Clifford group  is  a unitary 3-design. It is also a minimal 3-design except for dimension 4. As a consequence, any orbit of pure states of the multiqubit Clifford group is a 3-design; in particular, the set of multiqubit stabilizer states is a 3-design. The methods and  conclusions presented here are also useful in
studying higher moments of the Clifford group. These results 
are of interest to many research areas both within and beyond quantum information science.

Moreover, we  offered a simple explanation of why no discrete Wigner function is covariant with respect to the multiqubit Clifford group by proving that no operator basis is covariant with respect to any group that forms a unitary 3-design.
This result reveals a surprising connection between unitary designs and the physics of discrete phase spaces, which  is of  interest to studying quantum computation and a number of  nonclassical phenomena, such as negativity and contextuality.

Note added. Upon completion of this work, we noticed a comprehensive  math paper by Robert M. Guralnick and Pham Huu Tiep \cite{GuraT05}, from which it is possible to deduce our \thsref{thm:Clifford3design} and~\ref{thm:Clifford3designM} with some additional work. However, this paper mentions neither   $t$-designs nor the Clifford group explicitly.   In addition, some of their results rely on  Hering's theorem, which relies on the classification of finite simple groups (CFSG). Our proofs are completely independent of the CFSG and are thus simpler and more transparent. Recently (Sep 2015), unaware of our work (our  draft without \thref{thm:Clifford3designM} was completed in May 2015 and shared with a number of experts in the field), Zak Webb also proved  that the multiqubit Clifford group is a unitary 3-design (published by now \cite{Webb16}), which offers a complementary perspective to our approach.

\acknowledgments
The author is grateful to David Gross,   Richard Kueng, Debbie Leung, and Andreas Winter
for discussions.
This work was supported in part by Perimeter Institute for Theoretical Physics.
Research at Perimeter Institute is supported by the Government of Canada
through Industry Canada and by the Province of Ontario through the Ministry
of Research and Innovation.  The author also acknowledges financial support
from the Excellence
Initiative of the German Federal and State Governments
(ZUK~81) and the DFG. 

\appendix

\section{Proof of \lref{lem:CliffordFP}}
\begin{proof}According to \lref{lem:FPformula}, the frame potential $\Phi_t(\pc)$ is equal to the number of orbits of $\Sp{2n}{p}$ on  $(\bbF_p^{2n})^{\times (t-1)}$. The number is 2 when $t=2$ given that $\Sp{2n}{p}$ is transitive \cite{DixoM96book, Came99book}. When $t=3$, let $\mathbf{0}=(0,0,\ldots, 0)^\rmT$ and $\mathbf{1}=(1,0,\ldots, 0)^\rmT\in \bbF_p^{2n}$. Then any orbit on $(\bbF_p^{2n})^{\times 2}$ contains one of the following elements $(\mathbf{0}, \mathbf{0})$, $(\mathbf{0}, \mathbf{1})$, $(\mathbf{1}, \mathbf{0})$, $(\mathbf{1}, \mathbf{1})$, and $(\mathbf{1}, \mu)$, where $\mu\neq \mathbf{0}, \mathbf{1}$. The vector $\mu$ is a fixed point of the stabilizer of $\mathbf{1}$ if and only if it is proportional to $\mathbf{1}$; there are $p-2$ such fixed points excluding $\mathbf{0}, \mathbf{1}$.  Suppose  $\mu,\nu\in\bbF_p^{2n*}$ are not proportional to $\mathbf{1}$, then $(\mathbf{1}, \mu)$ and $(\mathbf{1}, \nu)$ are on the same orbit if and only if the symplectic products $\langle \mathbf{1}, \mu\rangle$ and $\langle \mathbf{1}, \nu\rangle$ are equal by Witt's lemma \cite{Came00book}. When $n>1$, $\langle \mathbf{1}, \mu\rangle$ may take on any value in $\bbF_p$, while it is nonzero when $n=1$. So there are $2p+1$ orbits in total when $n=1$ and $2p+2$ orbits  when $n>1$, from which we deduce \eref{eq:CliffordFP3}. \Eref{eq:CliffordFP4} and  frame potentials of higher orders can be derived using a similar reasoning.

	For the restricted Clifford group, the summation over $f(F)^{t-1}$ in \eref{eq:FPformula} can be evaluated explicitly.  \Eref{eq:RCliffordFP} follows from the fact that $f(F)=q$ for the $q^2-1$
	order-$p$ elements in $\SL{2}{q}$ and $f(F)=1$ for other $q^3-q^2-q$ nonidentity elements  (see  \rscite{Hump75,ApplBC08, Zhu10} for the conjugacy classes of  $\SL{2}{q}$). 
\end{proof}

\section{Proof of \lref{lem:FPformula}}

\begin{proof}
	Let $F\in R$ and $U_F$ be a Clifford unitary that induces the transformation $F$; then $U_F D_\mu$ induces the same transformation for all $\mu\in \bbF_p^{2n}$.
	According to a similar argument used in the proof of Theorem 2.34 in Zauner's thesis \cite{Zaun11} (see also \rcite{Zhu12the}), $|\tr(U_FD_\mu)|^2$ is either zero or equal to the number of  displacement operators that commute with $U_F$, which in turn is  equal to the number  $f(F)$ of fixed points of $F$ in $\bbF_p^{2n}$. On the other hand,  $\sum _\mu |\tr(U_F D_\mu)|=q^2$ given that the HW group is a unitary error basis.
	It follows  that $|\tr(U_FD_\mu)|^2=f(F)$ for $q^2/f(F)$ of the  $q^2$ displacement operators $D_\mu$. Therefore,
	\begin{align}
	&\Phi_t(\overline{G})=\frac{1}{q^2|R|} \sum_{F\in R} \sum _\mu |\tr(U_FD_\mu)|^{2t} \nonumber \\
	&=\frac{1}{q^2|R|}\sum_{F\in R} f(F)^{t}\frac{q^2}{f(F)}=\frac{1}{|R|} \sum_{F\in R} f(F)^{t-1}.
	\end{align}
	According to the orbit-stabilizer relation,  $\Phi_t(\overline{G})$ is equal to the number of orbits of $R$ on $(\bbF_p^{2n})^{\times (t-1)}$.

	In view of   \eref{eq:FramePmin} and \lref{lem:FPformula}, the group $\overline{G}$ is a unitary 2-design if and only if $R$ has two orbits  on $\bbF_p^{2n}$ and is transitive on  $\bbF_p^{2n*}$. The group $\overline{G}$
	is a unitary 3-design if and only if $R$ has five orbits on $(\bbF_p^{2n})^{\times 2}$ when $n=1$ and six orbits when $n\geq2$; that is,  $R$ has two orbits 
	on $(\bbF_p^{2n*})^{\times 2}$ and is  2-transitive when $n=1$, and it has three orbits  and rank-3 when $n\geq2$.
\end{proof}

\section{Proof of \thref{thm:Clifford3designM}}
\begin{proof}
	Suppose $\overline{G}$ is a subgroup of the $n$-qubit Clifford group that
	is a unitary 3-design. Let $\overline{H}=\overline{G}\phw$ and $R=\overline{H}/\phw$.
	Then $\overline{H}$ is also a unitary 3-design. By \lref{lem:FPformula},
	$R$ is 2-transitive when $n=1$ and has  rank-3  when $n\geq2$.
	According to \thref{thm:Rank3subSp} below, $R=\Sp{2n}{2}$ or $n=2$ and $R=A_6$ (the conclusion for $n=1$ follows from the observation that $\Sp{2}{2}\simeq S_3$). Here  $S_m$ and $A_m$ denote  the symmetric group and alternating group on $m$ letters. To complete the proof it remains to show that
	$\overline{G}\geq\phw$. Suppose this is not the case; then $\overline{G}$
	cannot contain any nontrivial displacement operator given that $R$ is transitive. Therefore, $\overline{G}$ is a complement
	of the HW group and is isomorphic to $\Sp{2n}{2}$  when
	$n\neq2$; when $n=2$, $\overline{G}$ is  isomorphic  to either  $\Sp{4}{2}$
	or $A_6$. However, any unitary 3-design in dimension $d$ has at least  $d^2(d^4
	- 3d^2 + 6)/2$ elements \cite{RoyS09}, which means 20 and 1712 elements for
	dimensions  2 and 4 and leads to a contradiction when $n=1,2$. In addition,
	the HW group  is not complemented in the $n$-qubit Clifford group when $n\geq2
	$ according to Theorem~7 in \rcite{BoltRW61II}. This contradiction confirms
	\thref{thm:Clifford3designM}. 
\end{proof}

\begin{theorem}\label{thm:Rank3subSp}
	The group $\Sp{2n}{2}$ with $n\geq2$ has no proper rank-3 subgroup except
	when $n=2$, in which case it has a unique proper rank-3 subgroup, that
	is, the alternating group $A_6$ embedded in $\Sp{4}{2}$.
\end{theorem}
This  theorem follows from the work of  Cameron and Kantor \cite{CameK79,CameK02},  although it is not easy to spot an explicit statement in these references. To prove this theorem, we need to introduce a few auxiliary concepts and results.
A subgroup of $\Sp{2n}{2}$ is \emph{primitive} if it acts transitively on nonzero vectors in  $\bbF_2^{2n}$ and preserves no nontrivial partition.
It is \emph{antiflag transitive} if it acts transitively
on  all pairs $(\mu, H)$ of nonzero vectors and hyperplanes in $\bbF_2^{2n}$ with $\mu\notin H$ \cite{CameK79,CameK02}.
\begin{lemma}
	Any rank-3 subgroup of $\Sp{2n}{2}$ for  $n\geq2$ is primitive. 
\end{lemma}
\begin{proof}
	Let $R$ be a rank-3 subgroup of $\Sp{2n}{2}$; then $R$ is transitive, and
	the point stabilizer of any nonzero vector $\mu$ partitions the remaining
	nonzero vectors into two orbits according to their symplectic products with
	$\mu$. Therefore, the stabilizer has two orbits on $\bbF_2^{2n*}$ of lengths $2^{2n-1}-2$ and $2^{2n-1}$, respectively. If $R$ is not primitive, then  any block in a nontrivial partition has size either $2^{2n-1}-1$ or $2^{2n-1}+1$. On the other hand, the  size must be a divisor of $2^{2n}-1$. This contradiction shows that $R$ is primitive.
\end{proof}

\begin{lemma}
	Any rank-3 subgroup of $\Sp{2n}{2}$ for $n\geq2$  is antiflag transitive. 
\end{lemma}
\begin{proof}
	Let $R$ be a rank-3 subgroup; then $R$ is transitive and the point stabilizer of any nonzero vector $\mu$ has two orbits on the remaining nonzero vectors. 
	Denote by $\mu^\bot$ the hyperplane  composed of all vectors that are orthogonal to $\mu$ with respect to the given symplectic product. Then the map $\mu\mapsto\mu^\bot$ sets a one-to-one correspondence between vectors and hyperplanes, which is preserved by the symplectic group, that is,  $F\mu^\bot=(F\mu)^\bot$ for any $F\in\Sp{2n}{2}$. Let $\nu_1^\bot$ and $\nu_2^\bot$ be two hyperplanes that do not contain $\mu$, that is $\langle \mu, \nu_1\rangle=\langle \mu, \nu_2\rangle=1$. Then the point stabilizer of $\mu$ within $R$ can map $\nu_1$ to $\nu_2$ and, accordingly, $\nu_1^\bot$ to $\nu_2^\bot$. So $R$ is antiflag transitive. 
\end{proof}

\begin{proof}[Proof of \thref{thm:Rank3subSp}]
	Let $R$ be a  rank-3 subgroup of $\Sp{2n}{2}$ with $n\geq2$; then $R$ is primitive and antiflag transitive, but not 2-transitive.  Now \thref{thm:Rank3subSp}  follows from Theorem~5.2 and Proposition~6.2 in \rcite{CameK79}. Alternatively, it can be proved based on
	Theorem~2.2 in \rcite{CameK02}.
\end{proof}

\section{Proof of \thref{lem:u3design}}
\begin{proof}
	In view of \thref{thm:Clifford3design}, it suffices to prove the  statement that
	no operator basis is covariant with respect to a unitary group 3-design.
	Suppose on the contrary that $\{L_j\}$ is an operator basis on the Hilbert space  $\caH$ of dimension $d$ that is covariant with respect to a unitary group 3-design $\overline{G}$.
	Then $\Phi_2(\overline{G})=2$ and $\Phi_3(\overline{G})=6$ ($\Phi_3(\overline{G})=5$ when $d=2$) according to \eref{eq:FramePmin}. Note that $\{L_j\otimes L_k\}$ and $\{L_j\otimes L_k\otimes L_l\}$ form operator bases for $\caH^{\otimes2}$ and $\caH^{\otimes3}$, respectively. 
	According to Lemma 1 in \rcite{Zhu16P} (cf.~Lemma~7.2 in \rcite{Zhu12the}), $\overline{G}$ acts  transitively
	on ordered pairs of distinct operators in $\{L_j\}$ and has two orbits (one orbit when $d=2$) on ordered triples. The triple products $\tr(L_jL_kL_l)$ for distinct $j,k,l$ must all be equal and thus real when  $d=2$, while they can take on at most two different values when $d\geq3$.   
	
	However, these triple products cannot all be real since, otherwise, the basis operators would commute with each other and thus cannot form an operator basis.
	When $d=2$, this contradiction confirms the theorem. When $d\geq 3$,
	these triple products
	must take on two distinct values, which are complex conjugates of each other. Consequently,  $\overline{G}$ acts transitively on unordered triples; in other words, $\overline{G}$ is 3-homogeneous in the language of permutation groups \cite{DixoM96book,Came99book, Zhu15S}.
	According to Theorem 1 of Kantor~\cite{Kant72} (see also Theorem 9.4B in \rcite{DixoM96book} and  Lemma~2 in \rcite{Zhu15S}),  any 3-homogeneous permutation group  on  $m$ objects with $m\geq9$ a perfect square is $3$-transitive. Therefore, $\overline{G}$ acts transitively on ordered triples,  which means all triple products  $\tr(L_jL_kL_l)$ are real, in contradiction with the previous observation.
\end{proof}

\nocite{apsrev41Control}

\bibliographystyle{apsrev4-1}

\bibliography{all_references}

\end{document}